\label{key}
\documentclass[conference]{IEEEtran}
\IEEEoverridecommandlockouts
\pdfoutput=1
\usepackage{floatrow}
\usepackage{subcaption}
\usepackage{colortbl}
\usepackage{bm}
\usepackage{graphicx}
\usepackage{epstopdf}
\usepackage{enumerate}
\usepackage{relsize}

\usepackage{graphicx}  
\usepackage{url}       
\usepackage{geometry}
\usepackage{amsmath}   
\usepackage{cleveref}
\usepackage{cite}
\usepackage{extarrows}
\usepackage{amsfonts,amssymb}

\usepackage{stfloats}

\usepackage{cases}
\usepackage{algorithm}
\usepackage{multirow}
\usepackage{algorithmic}
\usepackage[T1]{fontenc}
\usepackage{amsfonts}
\allowdisplaybreaks[4]

\newtheorem{lemma}{{Lemma}}
\newtheorem{proof}{{Proof}}

\graphicspath{{Matlab/}{Visio/}}

\begin{document}
\title{Throughput and Fairness Trade-off Balancing for UAV-Enabled Wireless Communication Systems}

\author{\IEEEauthorblockN{Kejie Ni$^{\dagger}$, Jingqing Wang$^{\dagger}$, Wenchi Cheng$^{\dagger}$, and Wei Zhang$^{\ddagger}$}~\\[0.2cm]
\vspace{-10pt}

\IEEEauthorblockA{$^{\dagger}$State Key Laboratory of Integrated
Services Networks, Xidian University, Xi'an, China\\
$^{\ddagger}$School of Electrical Engineering and Telecommunications, The University of New South Wales, Sydney, Australia\\
E-mail: \{\emph{kjni@stu.xidian.edu.cn}, \emph{jqwangxd@xidian.edu.cn}, \emph{wccheng@xidian.edu.cn}, \emph{w.zhang@unsw.edu.au}\}}

\vspace{-20pt}
}

\maketitle

\begin{abstract}
Given the imperative of 6G networks' ubiquitous connectivity, along with the inherent mobility and cost-effectiveness of unmanned aerial vehicles (UAVs), UAVs play a critical role within 6G wireless networks. Despite advancements in enhancing the UAV-enabled communication systems' throughput in existing studies, there remains a notable gap in addressing issues concerning user fairness and quality-of-service (QoS) provisioning and lacks an effective scheme to depict the trade-off between system throughput and user fairness. To solve the above challenges, in this paper we introduce a novel fairness control scheme for UAV-enabled wireless communication systems based on a new weighted function. First, we propose a throughput combining model based on a new weighted function with fairness considering. Second, we formulate the optimization problem to maximize the weighted sum of all users' throughput. Third, we decompose the optimization problem and propose an efficient iterative algorithm to solve it. Finally, simulation results are provided to demonstrate the considerable potential of our proposed scheme in fairness and QoS provisioning. 
\end{abstract}

\vspace{10pt}

\begin{IEEEkeywords}
UAV communications, weighted function, user fairness, QoS, resource allocation, trajectory design.
\end{IEEEkeywords}

\section{Introduction}
\IEEEPARstart{W}{ith} the rapid development of novel technologies, the demand for high-speed and reliable wireless communications has escalated. The future 6G networks aim to provide ubiquitous connectivity and massive connectivity. Although existing ground communication facilities generally meet the daily communication demands, they are insufficient when faced with unexpected situations, such as network reconstruction after natural disasters, temporary communication deployments in remote areas, and wireless resource allocation during large-scale gatherings or holiday events~\cite{UAVsituation}. With advantages including high mobility, rapid deployment and cost-effective, Unmanned Aerial Vehicles (UAVs) can be deployed to provide auxiliary communication in these scenarios. Therefore, UAVs play a vital role in the 6G wireless networks~\cite{6GUAV}.  

Despite the promising potential of UAVs in 6G wireless networks, there still exists a number of challenges in optimizing UAV-enabled communication system's throughput while guaranteeing user fairness and quality-of-service (QoS). In~\cite{emergency}, optimal resource allocation scheme is proposed for UAV-based emergency wireless communications. Further considering statistical QoS ptovisionings, the authors in~\cite{qos} propose an efficient framework to jointly optimize spectrum and power efficiencies over SISO/MIMO wireless networks. A UAV-enabled OFDM system is studied~\cite{commonthroughput}~\cite{joint}, where the UAV is dispatched as the mobile base station (BS) to serve a group of ground users. Considering user fairness and QoS, the authors maximizes the minimum average throughput of all users by jointly optimizing resource allocation and trajectory design. Different from maximizing minimum throughput or energy efficiency, Jain's fairness index is also widely adopted to ensure user fairness~\cite{Jain1}~\cite{Jain3}. The smaller the differences of energy efficiency, coverage efficiency or throughput among users are, the greater Jain's fair index is. 

In addition, there exists a trade-off between system throughput and user fairness. However, both Jain's fairness index and maximizing minimum throughput of all users overlook system throughput. The cost of system throughput for fairness is rather large, especially when there are users with particularly poor channel conditions or those located in extremely remote areas. Therefore, there is a lack of a flexible mechanism that can adjust the scale of this trade-off. 

In light of the aforementioned challenges, it is imperative to develop the mobility and application potential of UAVs in next-generation wireless networks. In this paper, we propose a trade-off control mechanism for UAV-enabled wireless communication systems that can adjust the scale of the trade-off between system throughput and user fairness. We first propose throughput combining model based on a new weighted function. Next, we formulate the optimization problem to maximize the weighted sum function of all users' throughput. According to the convex condition of the weighted function, the original optimization problem is decomposed into two different structures, which are further decomposed into two separate subproblems. Due to their non-convexity, we reconstruct them as convex problems and propose an efficient iterative algorithm by employing Lagrange dual, slack variables, and Taylor approximation. Finally, comprehensive simulation results are provided to demonstrate the considerable potential of our proposed scheme in fairness and QoS provisioning. 

The rest of this paper is organized as follows. In Section \ref{Sec:System}, we introduce the system model, properties of our proposed weighted function and the problem formulation. In Section \ref{Sec:Sol}, we propose an efficient iterative optimization algorithm. In Section \ref{Sec:Results}, we provide numerical results to demonstrate the desirable properties and performance of our proposed scheme. Finally, we conclude the paper in Section \ref{Sec:conclusion}.
\begin{figure}[!t]
	\centering  
	\includegraphics[width=0.8\columnwidth]{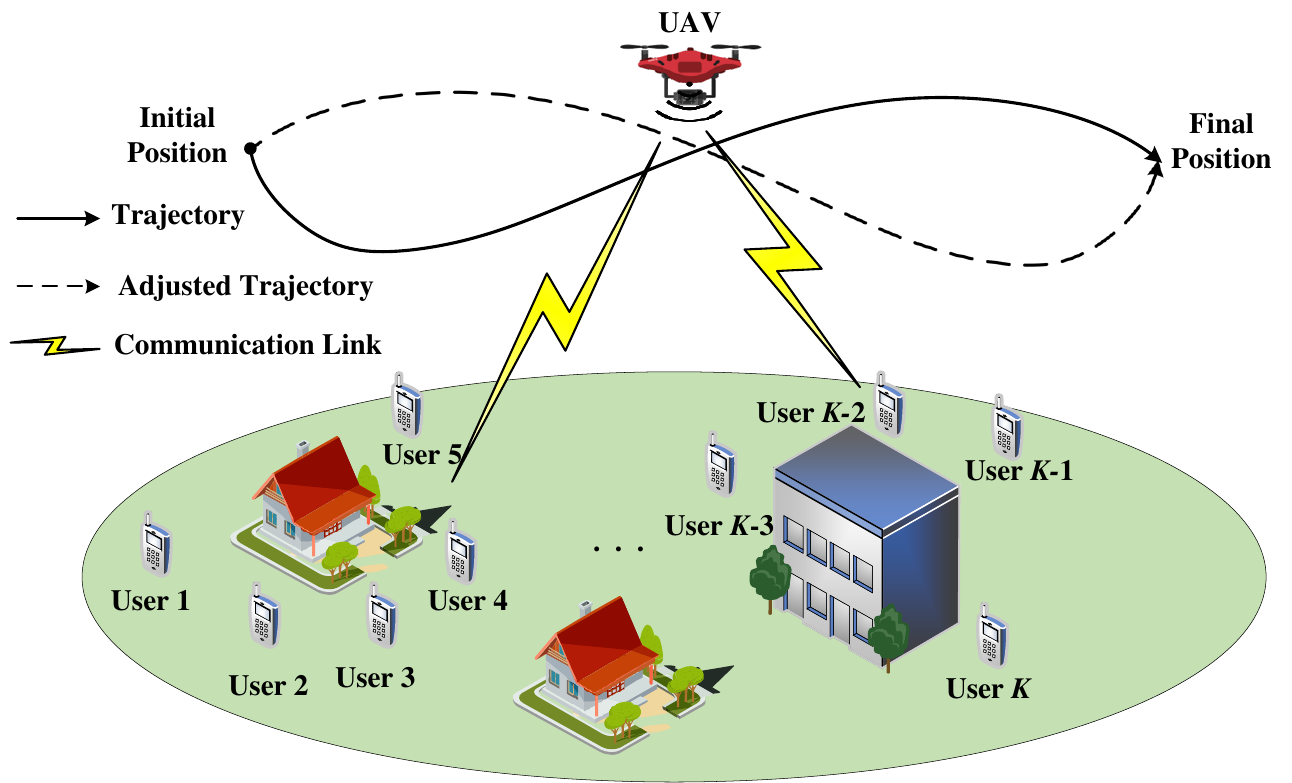}  
	\caption{UAV-enabled wireless network}  
	\label{fig0}  
\end{figure}
\section{System Model and Problem Formulation}\label{Sec:System}
\subsection{System Model}
As shown in Fig.\ref{fig0}, we consider a UAV-enabled OFDM wireless communication system where \emph{K} ground users can be served simultaneously by one UAV. The 3-D Cartesian coordinates of ground user \emph{k} $\in$ $\mathcal{K}$ = $\left\{1,2,...,K\right\}$ is assumed as $\left[\mathbf{w}_k^T,0\right]$, where $\mathbf{w}_k$ = $\left[x_k,y_k\right]^T\in\mathbb{R}^{2\times1}$. We discretize UAV flight period \emph{T} into \emph{N} time slots indexed by \emph{n} $\in$ $\mathcal{N}$ = $\left\{1,2,...,\emph{N}\right\}$, each of which has a equal length denoted by $\delta_t$ = $\frac{T}{N}$. The UAV flies at a fixed altitude $H$ with a given maximum speed $V_{max}$. The UAV trajectory is characterized by a sequence of UAV locations \scalebox{0.8}{$\left[\mathbf{q}\left[n\right]^T,H\right]^T$} , where $\mathbf{q}$$\left[n\right]$ = $\left[x\left[n\right],y\left[n\right]\right]^T$ $\in$  $\mathbb{R}^{2\times1}$.

We adopt the elevation-angle-dependent Rician fading channel model to model the communication links between the UAV and ground users, which is more accurate than free space fading channel. The channel power gain from the UAV to user $k$ in time slot $n$ can thus be expressed as
\begin{equation}
	\small
	h_k[n]=h_0f_k[n]d^{-2}_k[n]=\frac{h_0f_k[n]}{H^2+||\mathbf{q}[n]-\mathbf{w}_k||^2}
\end{equation}
where $h_0$ denotes the channel power gain at the reference distance $d_0=1$ m, and $d_k[n]=\sqrt{H^2+||\mathbf{q}[n]-\mathbf{w}_k||^2}$ is the distance between user $k$ and the UAV in time slot n. Here, $f_k[n]$ is the effective fading power of the channel, which can be approximated as~\cite{Rician_approximation}
\begin{equation}
	\small
	f_k[n]=C_1+\frac{C_2}{1+\mathrm{exp}(-(B_1+B_2\theta_k[n]))}
\end{equation}
where $C_1, C_2, B_1$ and $B_2$ are all constants determined by channel conditions, and $\theta_k[n] = \frac{H}{\sqrt{H^2+||\mathbf{q}[n]-\mathbf{w_k}||^2}}$.

As a result, the achievable throughput of user $k$ in time slot $n$ denoted by $R_k[n]$ in bits/second/Hz (bps/Hz) is given by
\begin{equation}
	\small
	\begin{aligned}
	R_k[n]=b_k[n]\mathrm{log}_2(1+\frac{p_k[n]Ph_k[n]}{b_k[n]BN_0})
	\end{aligned}
\end{equation}
where $N_0$ denotes the power spectral density of the addictive white Gaussian noise (AWGN) at the receivers, $b_k[n]$ and $p_k[n]$ are bandwidth and power allocation coefficients with regard to available power $P$ and bandwidth $B$.

In order to depict and control the difference of users' instantaneous throughput, inspired by~\cite{anovelrate}, we propose a new weighted sum function of instantaneous throughput in the following form

\begin{equation}
	\small
	C=\sum_{k=1}^{K} \omega_kR_k[n]\label{con:weighted sum}
\end{equation} 
where
\begin{equation}
	\small
	\omega_k=\frac{e^{-\alpha{R_k[n]}}}{\sum_{i=1}^{K} e^{-\alpha{R_i[n]}}}.
	\label{weighted function}
\end{equation}

The proposed weighted sum function is able to compensate those users with lower throughput resulted from unfair resource allocation and UAV trajectory scheme. The value of fairness factor $\alpha$ determines how much the compensation is made. In other word, we can flexibly adjust the degree of the trade-off by tuning the value of $\alpha$.

Before problem formulation, we define a function of $n$ variables $x\triangleq[x_1,\dots,x_n],x_i>0,i= 1,\dots,n$, parameterized by $\alpha\geq0$ as
\begin{equation}
	\small
	H_{\alpha}(x)=\frac{\sum_{j=1}^{K}{x_j}e^{-\alpha{x_j}}}{\sum_{i=1}^{K} e^{-\alpha{x_i}}}.\label{Hx}
\end{equation}
Next we will show several properties of the weighted sum function (\ref{con:weighted sum}) by the following two lemmas.
\begin{lemma} 
	\emph{The function $H_\alpha(x)$ given by (\ref{Hx}) is concave and non-decreasing for $\alpha{x_i}\leq1,\forall i$}.
	\label{lemma1}
\end{lemma}
\begin{proof}
	See Appendix \ref{appendix 1}.
\end{proof}
\begin{lemma} 
	\emph{For $x_i\geq 0,\forall i$, as $\alpha$ $\to$ $+\infty$, $H_{\alpha}(x)$ simplifies to}
	\begin{equation}
		\small
		H_\infty (x)\triangleq \lim\limits_{\alpha \to +\infty} H_\alpha (x)=\underset{j}\min x_j. 
	\end{equation}
	\label{lemma2}
\end{lemma}
\vspace{-10pt}
\begin{proof}
	See Appendix B.
\end{proof}

Remarks on Lemma 1 and Lemma 2: Lemma 1 and Lemma 2 show the convex condition of the weighted sum function (\ref{con:weighted sum}) and facilitate the decomposition of the optimization problem into two different structures in the subsequent problem formulation.
\subsection{Problem Formulation}

Denote bandwidth allocation
$\mathbf{B}$ = $\left\{b_k[n],k\in\mathcal{K},n\in\mathcal{N}\right\}$, power allocation $\mathbf{P} = \left\{ p_k[n] \mid k \in \mathcal{K}, n \in \mathcal{N} \right\},$ and trajectory design  
$\mathbf{Q} = \left\{ \mathbf{q}[n], n \in \mathcal{N} \right\}.$ We optimize average weighted sum of $K$ users' throughput over $N$ time slots by jointly optimizing the bandwidth and power allocation, as well as the UAV trajectory. 

According to the convex condition of the weighted sum function shown in Lemma 1 and Lemma 2, we decompose the optimization problem into two different structures, which can be formulated as follows.
\subsubsection{Condition 1: for $\alpha R_k[n]\leq1,\forall {k,n}$ (This condition can be satisfied by minimizing the quantization interval of $\alpha$ no matter how large $R_k[n]$ is, which will be expressed later in the simulation results analysis)}
\begin{subequations}  
	\small
	\begin{align}  
		\max_{\mathbf{B},\mathbf{P},\mathbf{Q}} \quad & \frac{1}{N}\sum_{n=1}^{N}\sum_{k=1}^{K}\frac{e^{-\alpha R_k[n]}}{\sum_{i=1}^{K}e^{-\alpha R_i[n]}}R_k[n] \label{conf:objective function1} \\  
		\text{s.t.} \quad &\sum_{k=1}^{K}b_k[n]\leq1, \sum_{k=1}^{K}p_k[n]\leq1\quad\forall n\label{band1}\\ 
		&0\leq b_k[n]\leq1, 0\leq p_k[n]\leq1\quad\forall k,n.\label{band2}\\
		& \mathbf{q}[1]=\mathbf{q}_I, \mathbf{q}[N]=\mathbf{q}_F \label{conf:trajectory constraint1} \\  
		& ||\mathbf{q}[n+1]-\mathbf{q}[n]||^2\leq V_\mathrm{max}^2\delta_t^2,\quad n=1,2,\dots,N-1 \label{conf:trajectory constraint2}  
	\end{align}  
	\label{conf:optimization problem1}  
\end{subequations} 
where $R_k[n]=b_k[n]\mathrm{log}_2 \left\{ 1+\left[C_1+\frac{C_2}{1+\mathrm{exp}(-(B_1+B_2\theta_k[n]))}\right]\right. \\ 
\left. \times \frac{p_k[n]\gamma_0}{b_k[n]d_k^2[n]} \right\},\theta_k[n] = \frac{H}{\sqrt{H^2+||\mathbf{q}[n]-\mathbf{w}_k||^2}}$, $\gamma_0 \triangleq \frac{Ph_0}{BN_0}$, $\mathbf{q}_I$ and $\mathbf{q}_F$ are the UAV's initial and final positions. Constraints $(\ref{band1})-(\ref{band2})$ are resource allocation constraints. The trajectory constraints of the USA are expressed by $(\ref{conf:trajectory constraint1})-(\ref{conf:trajectory constraint2})$.
\subsubsection{Condition 2: for $\alpha \to +\infty$ (It's a maximizing minimum instantaneous throughput problem)}
\begin{subequations}
	\small  
	\begin{align}  
		\max_{\mathbf{B},\mathbf{P},\mathbf{Q}} \quad & \frac{1}{N}\sum_{n=1}^{N}\eta_n \\  
		\text{s.t.} \quad & (\ref{band1}),(\ref{band2}),(\ref{conf:trajectory constraint1}),(\ref{conf:trajectory constraint2})\\
		&R_k[n]\geq \eta_n,\quad\forall k,n  
	\end{align}  
	\label{conf:maxmin problem1}
\end{subequations}
where $\eta_n\triangleq \underset{k\in\mathcal{K}}\min R_k[n],\forall n$.

\begin{algorithm}[t!]
	\caption{Joint Optimization of Bandwidth Allocation, Power Allocation and UAV Trajectory for $\alpha R_k[n]\leq1,\forall {k,n}$}
	\label{algorithm}
	\small
	\begin{algorithmic}[1] 
		\STATE Initialize $\mathbf{B}^0, \mathbf{P}^0, \mathbf{Q}^0$, $r_{\mathrm{max}}$, $\epsilon$ and let $r=0$.
		\STATE Calculate the initial objective value $V^0$ \text{w.r.t.} $\mathbf{B}^0, \mathbf{P}^0, \mathbf{Q}^0$.
		\REPEAT
		\STATE $r=r+1$.
		\STATE Given $\mathbf{Q}^{r-1}$, obtain $\mathbf{B}^{r}$ and $\mathbf{P}^{r}$ by solving problem $(\ref{conf:optimization problem2})$.
		\STATE Given $\mathbf{B}^{r}$, $\mathbf{P}^{r}$ and $\mathbf{Q}^{r-1}$, obtain $\mathbf{Q}^{r}$ by solving problem $(\ref{conf:optimization problem5})$.
		\STATE Given $\mathbf{B}^{r}$, $\mathbf{P}^{r}$ and $\mathbf{Q}^{r}$, obtain the objective Value $V^r$ for the $r$th iteration.
		\UNTIL $|V^r-V^{r-1}|\leq\epsilon$ or $r\geq r_{\mathrm{max}}$.
		\STATE $\mathbf{B}^{r}$, $\mathbf{P}^{r}$ and $\mathbf{Q}^{r}$ are the obtained optimal solution.		
	\end{algorithmic}
	\label{algorithm1}
\end{algorithm}
\addtolength{\topmargin}{0.141in}
\section{Joint Resource Allocation and Trajectory Design Algorithm}\label{Sec:Sol}    
Problem $(\ref*{conf:optimization problem1})$ and $(\ref{conf:maxmin problem1})$ are non-convex problems due to the lack of joint concavity of $R_k[n]$ with respect to the optimization variables $\mathbf{B},\mathbf{P}$ and $\mathbf{Q}$. Under both two conditions, we decompose the original optimization problem into several optimization subproblems and solve them iteratively until convergence is achieved.
\subsection{Condition 1: for $\alpha R_k[n]\leq1,\forall {k,n}$}
\subsubsection{Joint Bandwidth and Power Allocation}
We jointly optimize the bandwidth allocation $\mathbf{B}$ and power allocation $\mathbf{P}$ for any given UAV trajectory $\mathbf{Q}$. Hence, the problem $(\ref{conf:optimization problem1})$ can be reformulated as
\begin{subequations}  
	\small
	\begin{align}  
		\max_{\mathbf{B},\mathbf{P}} \quad & \frac{1}{N}\sum_{n=1}^{N}\sum_{k=1}^{K}\frac{e^{-\alpha R_k[n]}}{\sum_{i=1}^{K}e^{-\alpha R_i[n]}}R_k[n] \label{conf:objective function2_1} \\  
		\text{s.t.} \quad & (\ref{band1}),(\ref{band2})
	\end{align}  
	\label{conf:optimization problem2}  
\end{subequations}
where $R_k[n] \triangleq b_k[n]\mathrm{log}_2 ( 1+\gamma_k[n] \frac{p_k[n]}{b_k[n]})$,$\gamma_k[n] \triangleq (C_1+\frac{C_2}{1+\mathrm{exp}(-(B_1+B_2\theta_k[n]))})\frac{\gamma_0}{d_k^2[n]}$,$\forall k,n$. With the objective function and all constraints convex, the optimization problem (\ref{conf:optimization problem2}) is a convex problem, which can be solved by interior-point method.
\subsubsection{UAV Trajectory Optimization}
We optimize the UAV trajectory $\mathbf{Q}$ given specific bandwidth allocation $\mathbf{B}$ and power allocation $\mathbf{P}$. The optimization problem can be expressed as
\begin{subequations}  
	\small
	\begin{align}  
		\max_{\mathbf{Q}} \quad & \frac{1}{N}\sum_{n=1}^{N}\sum_{k=1}^{K}\frac{e^{-\alpha R_k[n]}}{\sum_{i=1}^{K}e^{-\alpha R_i[n]}}R_k[n] \label{conf:objective function3} \\  
		\text{s.t.} \quad &(\ref{conf:trajectory constraint1}),(\ref{conf:trajectory constraint2}) 
	\end{align}  
	\label{conf:optimization problem3}  
\end{subequations}  
where $R_k[n] \triangleq b_k[n]\mathrm{log}_2 \left[ 1+(C_1+\frac{C_2}{1+e^{-(B_1+B_2\theta_k[n])}})\right. \\
\left. \times \frac{\hat{\gamma}_k[n]}{H^2+||\mathbf{q}[n]-\mathbf{w}_k||^2}\right]$$,$$\hat{\gamma}_k[n]\triangleq\frac{\gamma_0p_k[n]}{b_k[n]}$,$\forall k,n$. The problem $(\ref{conf:optimization problem3})$ is non-convex due to the non-concavity of $(\ref{conf:objective function3})$ and non-convexity of constraint $(\ref{conf:trajectory constraint2})$ with respect to $\mathbf{q}[n]$. Since a convex function can be bounded by its first-order Taylor expansion, and $R_k[n]$ is convex with respect to $(1+e^{-(B_1+B_2\theta_k[n])})$ and $(H^2+||\mathbf{q}[n]-\mathbf{w}_k||^2)$~\cite{joint}. Therefore, we introduce a slack variable and adopt first-order Taylor approximation to transform the problem $(\ref{conf:optimization problem3})$ into a convex problem as follows
\begin{subequations}    
	\small
	\begin{align}    
		\max_{\mathbf{Q}} \quad & \frac{1}{N}\sum_{n=1}^{N}\sum_{k=1}^{K}\frac{e^{-\alpha \tilde{R}_k^{lb,r}[n]}}{\sum_{i=1}^{K}e^{-\alpha \tilde{R}_i^{lb,r}[n]}}\tilde{R}_k^{lb,r}[n] \label{conf:objective function5} \\    
		\text{s.t.} \quad & \Theta_k[n]\leq B_1+B_2\theta_k^{lb,r}[n], \quad \forall k, n \label{conf:slack variable2} \\    
		\quad &(\ref{conf:trajectory constraint1}),(\ref{conf:trajectory constraint2})    
	\end{align}    
	\label{conf:optimization problem5}    
\end{subequations}
where
\begin{equation}
	\small
	\begin{aligned}
        \tilde{R}^{lb,r}_k[n]
		\triangleq\tilde{R}_k^r[n]-\psi_k^r[n]\left(e^{-\Theta_k[n]}-e^{-\Theta_k^r[n]}\right)\\
		-\varphi_k^r[n]\left(||\mathbf{q}[n]-\mathbf{w}_k||^2-||\mathbf{q}^r[n]-\mathbf{w}_k||^2\right)
	\end{aligned}
\end{equation}
\vspace{-5pt}
\begin{equation}
	\small
	\tilde{R}_k^r[n]=b_k[n]\mathrm{log}_2\left[1+\left(C_1+\frac{C_2}{x_0}\right)\frac{\hat{\gamma}_k[n]}{y_0}\right]
\end{equation}
\vspace{-5pt}
\begin{equation}
	\small
	\psi_k^r[n]=b_k[n]\frac{C_2\hat{\gamma}_k[n]\mathrm{log}_2e}{x_0\left(x_0y_0+\left(C_1x_0+C_2\right)\hat{\gamma}_k[n]\right)}
\end{equation}
\vspace{-5pt}
\begin{equation}
	\small
	\varphi_k^r[n]=b_k[n]\frac{\left(C_1x_0+C_2\right)\hat{\gamma}_k[n]\mathrm{log}_2e}{y_0\left(x_0y_0+\left(C_1x_0+C_2\right)\hat{\gamma}_k[n]\right)}
\end{equation}
\vspace{-5pt}
\begin{equation}
	\small
	x_0=1+e^{-\Theta_k^r[n]},y_0=H^2+||\mathbf{q}^r[n]-\mathbf{w}_k||^2
\end{equation}
\vspace{-5pt}
\begin{equation}
	\scalebox{0.8}{$
	\begin{aligned}
		\theta_k^{lb,r}[n]\triangleq \frac{H}{\sqrt{H^2+||\mathbf{q}^r[n]-\mathbf{w}_k||^2}}-\frac{H}{2\left(H^2+||\mathbf{q}^r[n]-\mathbf{w}_k||^2\right)^{\frac{3}{2}}}\\
		\times\left(||\mathbf{q}[n]-\mathbf{w}_k||^2-||\mathbf{q}^r[n]-\mathbf{w}_k||^2\right)
    \end{aligned}$}
\end{equation}
where $\Theta_k[n]$ is a slack variable and $r$ is the iteration index.

Given that all constraints are convex, the problem $(\ref{conf:optimization problem5})$ is a convex problem and can be solved using the interior-point method.
\subsubsection{Overall Algorithm}
The original problem $(\ref{conf:optimization problem1})$ is divided into two subproblems $(\ref{conf:optimization problem2})$ and $(\ref{conf:optimization problem5})$ respectively. The overall algorithm for problem $(\ref{conf:optimization problem1})$ is shown in \textbf{Algorithm \ref{algorithm1}}, where $r_{\mathrm{max}}$ denotes the maximum iterations and $\epsilon$ denotes the tolerance error. Steps 2-7  involve alternately and iteratively solving problems $(\ref{conf:optimization problem2})$ and $(\ref{conf:optimization problem5})$ until the objective value converges or the maximum number of iterations is achieved.

\begin{figure}[!t]
	\centering  
	\includegraphics[width=0.8\columnwidth]{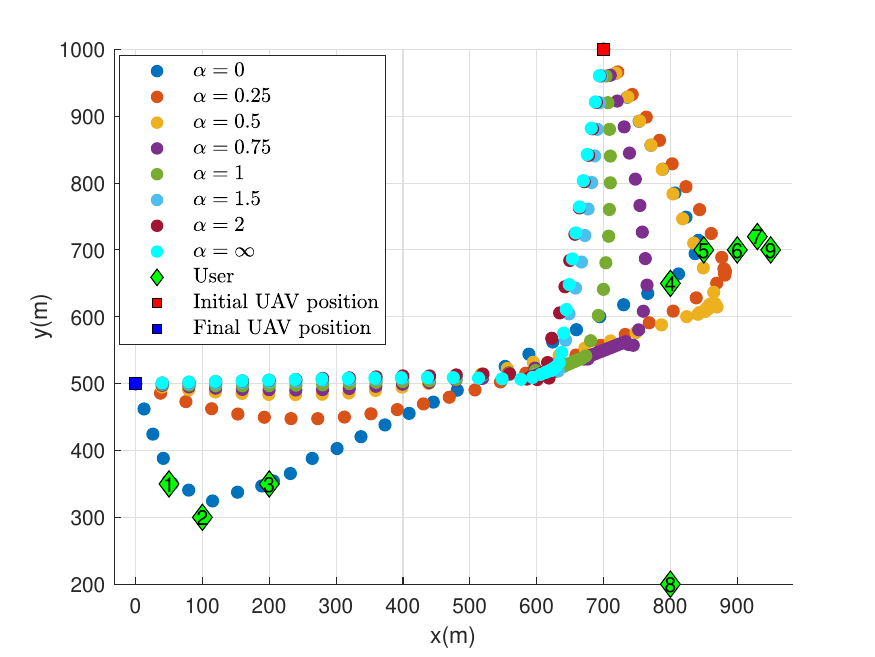}  
	\caption{Optimal trajectories under different values of fairness factor $\alpha$}  
	\label{fig1}  
\end{figure}

\subsection{Condition 2: for $\alpha \to +\infty$}
\subsubsection{Joint Bandwidth and Power Allocation}
The problem $(\ref{conf:maxmin problem1})$ can be reformulated as
\begin{subequations}  
	\small
	\begin{align}  
		\max_{\mathbf{B},\mathbf{P}} \quad & \frac{1}{N}\sum_{n=1}^{N}\eta_n \\  
		\text{s.t.} \quad & (\ref{band1}),(\ref{band2}) \\  
		&R_k[n]\geq \eta_n,\quad\forall k,n
		\label{conf:rate constraint}  
	\end{align}  
	\label{conf:maxmin problem2} 
\end{subequations}
where $R_k[n] = b_k[n]\mathrm{log}_2 ( 1+\gamma_k[n] \frac{p_k[n]}{b_k[n]})$,$\gamma_k[n] = (C_1+\frac{C_2}{1+\mathrm{exp}(-(B_1+B_2\theta_k[n]))})\frac{\gamma_0}{d_k^2[n]}$,$\forall k,n$. 
\subsubsection{UAV Trajectory Optimization}
The problem $(\ref{conf:maxmin problem1})$ can be reformulated as
\begin{subequations}  
	\small
	\begin{align}  
		\max_{\boldsymbol{\eta},\mathbf{Q}} \quad & \frac{1}{N}\sum_{n=1}^{N}\eta_n \\  
		\text{s.t.} \quad &(\ref{conf:trajectory constraint1}),(\ref{conf:trajectory constraint2})\\
		&R_k[n]\geq \eta_n,\quad\forall k,n.  
	\end{align}  
	\label{conf:maxmin problem3} 
\end{subequations}
Remark: Problems (\ref{conf:maxmin problem2}) and (\ref{conf:maxmin problem3}) are classic maximizing minimum problems~\cite{commonthroughput}~\cite{joint}, which can be solved by solving Lagrange dual, and employing Taylor approximation respectively. We omit the solutions due to the limitation of space.

\section{Numerical Results}\label{Sec:Results}
In this section, we evaluate the performance of our proposed scheme. We consider a system with $K=9$ ground users arbitrarily distributed on a horizontal plane. The UAV is assumed to fly at a fixed altitude $H=500$ m~\cite{commonthroughput} with a given maximum speed $V_\mathrm{max}=40$ m/s. The total available bandwidth is $B=10$ MHz and the noise power spectrum density is $N_0=-169$ dBm/Hz. The maximum transmission power is $P=0.1$ W and the channel power gain at the reference distance $d_0=1$ m is $h_0=-50$ dB. The Rician fading model parameters are given by $B_1=-4.3221$, $B_2=6.0750$, $C_1=0$ and $C_2=1$~\cite{joint}. The flight period and number of time slots are set as $T=50$s and $N=50$ respectively.

In Fig. \ref{fig1}, we compare the optimal trajectories of the UAV under different values of fairness factor $\alpha$. For clarity, the users are marked by `$\diamond$'s with their respective numbers inside. It can be observed that as the fairness factor $\alpha$ changes, the UAV adjusts its trajectories accordingly. When $\alpha$ increases, the UAV gradually shifts its trajectories from user-dense regions to the sparse one, to meet the fairness requirement.

In Fig. \ref{fig:f2}, we compare the throughput of all the users in different time slots with different values of fairness factor $\alpha$. It can be observed that there exists a `channel-quality bias' in resource allocation. When the fairness requirement is not strict (i.e. $\alpha$ is small), the UAV allocates bandwidth and power preferentially to users with better channel quality. With the increase of $\alpha$, the UAV gradually allocates more resources to users experiencing poorer channel quality. When $\alpha=0$, our scheme converges to a water-filling scheme, allocating maximal resources towards users with better channel quality. When $\alpha=\infty$, our scheme converges to channel inversion scheme, allocating resources fairly to all users and ensuring that each user's throughput remains consistent across any given time slot.
\begin{figure}[!t]  
	\centering  
	\begin{subfigure}[b]{0.48\linewidth} 
		\centering  
		\includegraphics[width=\linewidth]{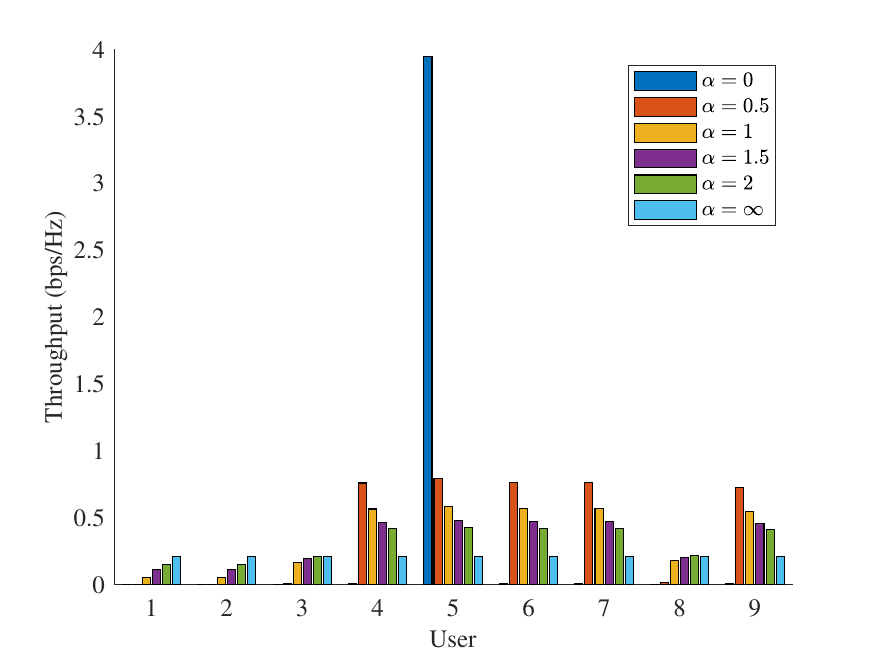}  
		\caption{Time slot $n=1$}  
		\label{fig:f2a}  
	\end{subfigure}  
	\hfill 
	\begin{subfigure}[b]{0.48\linewidth}  
		\centering  
		\includegraphics[width=\linewidth]{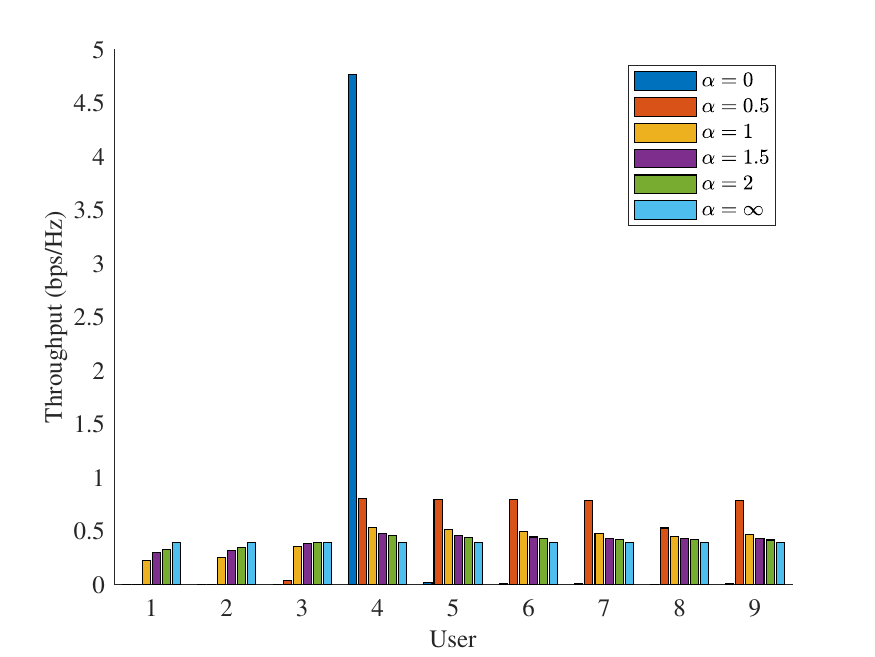}  
		\caption{Time slot $n=17$}  
		\label{fig:f2b}  
	\end{subfigure}  
	
	\begin{subfigure}[b]{0.48\linewidth}  
		\centering  
		\includegraphics[width=\linewidth]{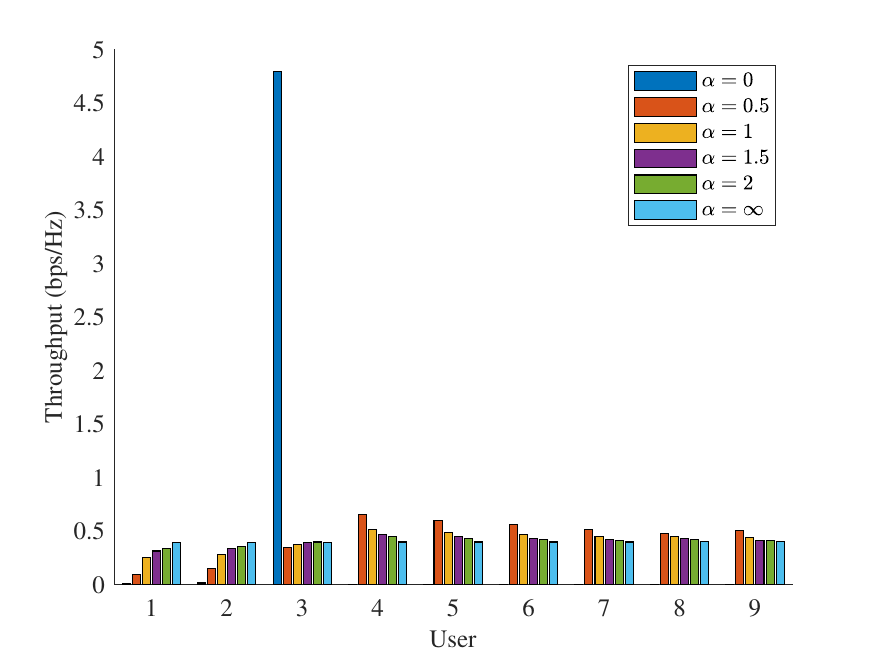}  
		\caption{Time slot $n=34$}  
		\label{fig:f2c}  
	\end{subfigure}  
	\hfill  
	\begin{subfigure}[b]{0.48\linewidth}  
		\centering  
		\includegraphics[width=\linewidth]{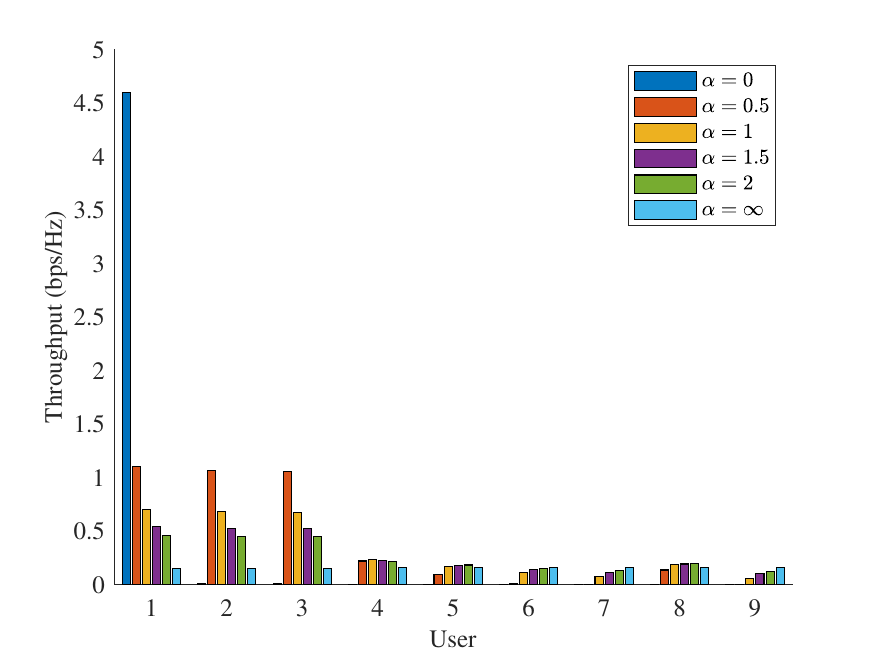} 
		\caption{Time slot $n=50$}  
		\label{fig:f2d}  
	\end{subfigure}  
	
	\caption{Throughput of all the users in different time slots with different values of fairness factor $\alpha$}  
	\label{fig:f2}  
\end{figure}

\begin{figure}[!t]    
	\centering    
	\begin{subfigure}[t]{0.48\linewidth}  
		\centering    
		\includegraphics[width=\linewidth]{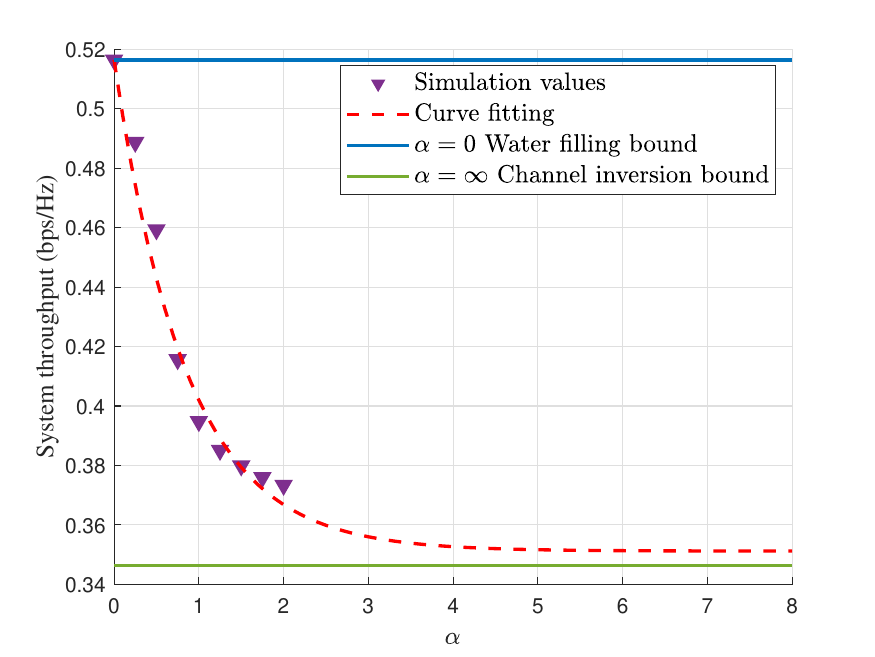} 
		\caption{System throughput versus fairness factor $\alpha$}    
		\label{fig:f3a}    
	\end{subfigure}    
	\hfill 
	\begin{subfigure}[t]{0.48\linewidth}  
		\centering    
		\includegraphics[width=\linewidth]{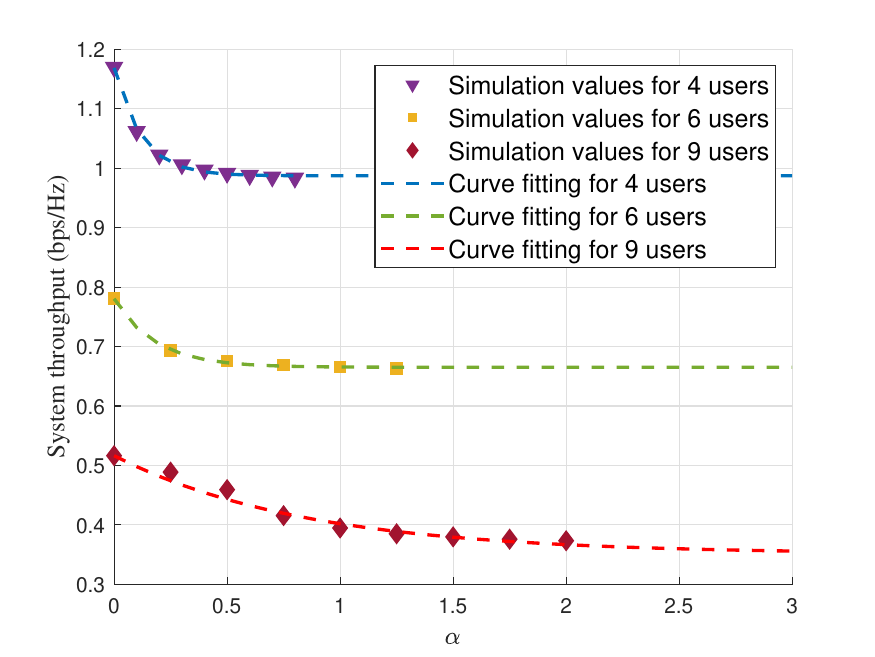} 
		\caption{System throughput versus fairness factor $\alpha$ with different numbers of users}    
		\label{fig:f3b}    
	\end{subfigure}    
	\caption{Comparison of system throughput versus fairness factor $\alpha$ in different scenarios}  
	\label{fig:f3}    
\end{figure}

In Fig. \ref{fig:f3} (a), we analyze system throughput (i.e. average throughput of $K$ users over $N$ time slots) of the developed UAV communication system. We select several values of $\alpha$ that ensure the convexity of the initial optimization problem (\ref{conf:optimization problem1}) and get the simulation values marked by `$\triangledown$'s. By plotting the fitting curve of these simulation values, we observe that the system throughput decays exponentially with respect to $\alpha$. Moreover, the system throughput is tightly bounded by the water-filling scheme and channel inversion scheme. This phenomenon reveals the trade-off between fairness and the system's throughput, which can be flexibly adjusted by tuning the value of fairness factor $\alpha$. In Fig. \ref{fig:f3} (b), we compare the system throughput with different numbers of ground users. We find that fewer users result in higher system throughput. However, higher throughput implies a smaller dynamic range of $\alpha$ that guarantees the convex constraints $\alpha R_k[n]\leq1$, which can be easily solved by minimizing the quantization interval of $\alpha$. It can also be observed that the curve decays faster with fewer users, which can be attributed to the numerator of the weighted function (\ref{weighted function}). By taking the derivative, we can ascertain that the numerator is more sensitive to the change of $\alpha$ with a higher value of throughput, indicating a faster system throughput decay with respect to $\alpha$.

\begin{figure}[!t]
	\centering
	\begin{subfigure}[t]{0.48\linewidth}
		\centering
		\includegraphics[width=\linewidth]{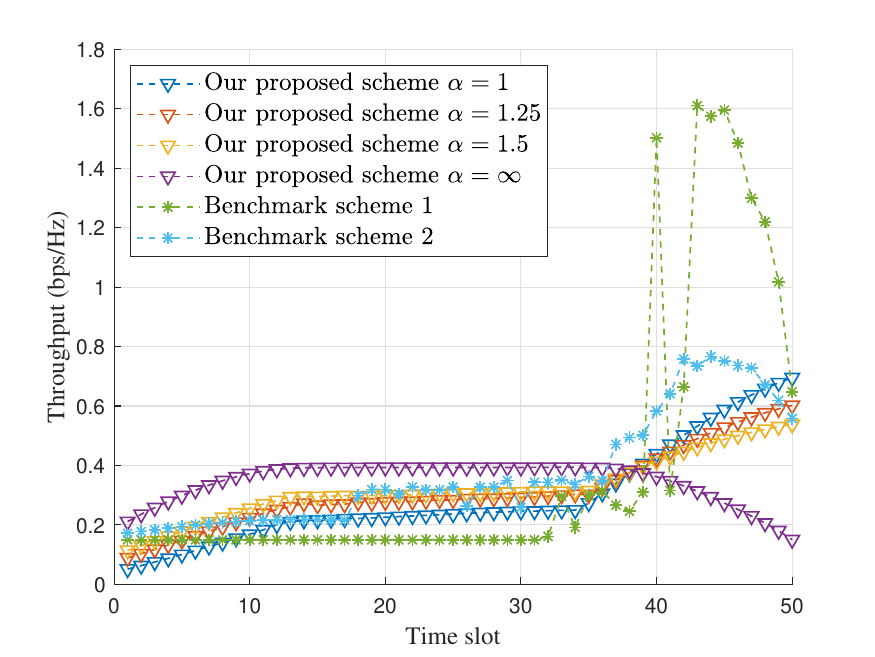}
		\caption{Throughput of our proposed scheme and benchmark scheme versus time slot with different factor values}
		\label{fig:f4a}
	\end{subfigure}
	\begin{subfigure}[t]{0.48\linewidth}
		\centering
		\includegraphics[width=\linewidth]{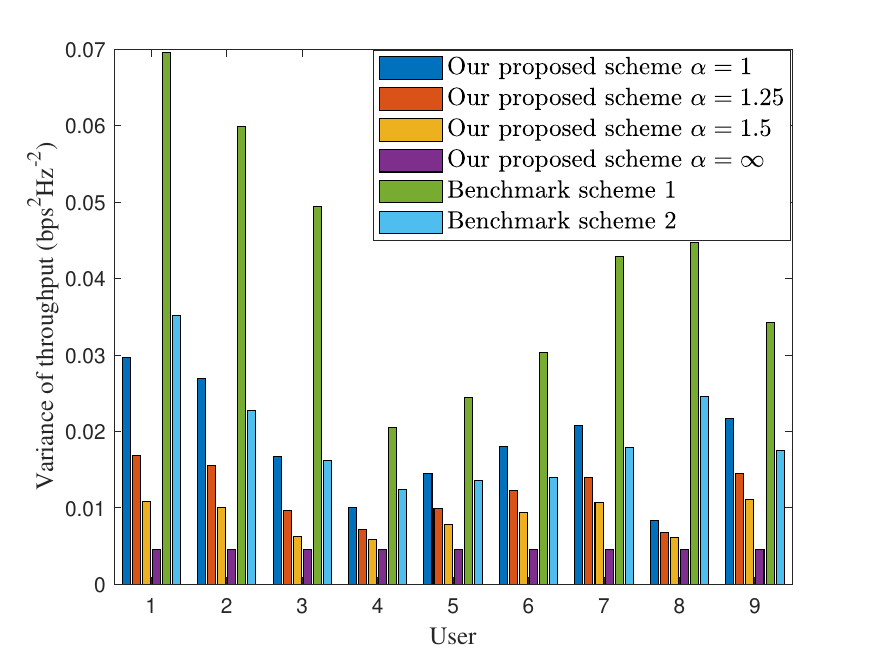}
		\caption{Variance of throughput of our proposed scheme and benchmark scheme among all users with different factor values}
		\label{fig:f4b}
	\end{subfigure}
	\caption{Comparison of performance between our proposed scheme and benchmark scheme}
	\label{fig:f4}
\end{figure}
In Fig. \ref{fig:f4}, we compare the performance of our proposed scheme and benchmark schemes with different fairness factor values. The benchmark schemes maximize the minimum average throughput of $K$ users (i.e. $\max\eta$, where $\eta=\min_{k\in\mathcal{K}}\frac{1}{N}\sum_{n=1}^{N}R_k[n]$). We consider two types of QoS constraints:$R_k[n]\geq R_{th}$ and $R_k[n]\geq\lambda\eta,\forall k,n$, which are referred to as `benchmark scheme 1' and `benchmark scheme 2' respectively. In the simulations, we adopt the most stringent QoS constraints for both benchmark schemes.
In Fig. \ref{fig:f4} (a), without loss of generality, we compare the throughput of `user 1' in all $N$ time slots between our proposed scheme and benchmark schemes with different fairness factor values. By enhancing the fairness requirement, our proposed scheme dynamically adjusts resource allocation, compensating users in poorer channel conditions and minimizing throughput fluctuations. We use the variance of throughput as a metric to quantify the stability of throughput as depicted in Fig. \ref{fig:f4} (b). Apparently, if relatively strict fairness requirement is imposed, our scheme is able to better guarantee the throughput stability than both benchmarks, even if they have employed the most stringent QoS constraints.

\section{Conclusion}\label{Sec:conclusion}
In this paper, we have proposed a trade-off control mechanism for UAV-enabled wireless communication systems based on a new weighted function. We formulated the optimization problem to maximize the weighted sum of all users' throughput. According to the convex condition of the weighted function, we decomposed the optimization problem into two different structures, which were further decomposed into two separate subproblems. Due to their non-convexity, we reconstructed them as convex problems and proposed an efficient iterative algorithm to solve them. Simulation results have verified our proposed scheme's flexible adjustment between system throughput and user fairness. Comparisons with benchmark schemes have demonstrated our proposed scheme's better performance in throughput stability and QoS provisioning.

\renewcommand{\appendixname}{Appendix~\Alph{section}}
\appendix[Proof of Lemma \ref{lemma1}]\label{appendix 1} 
We compute the first-order partial derivative of $H_\alpha$$(x)$, which is as follows
\begingroup
\smaller
\begin{align*}
		\frac{\partial H_\alpha (x)}{\partial x_j}=\frac{(1-\alpha x_j)e^{-\alpha x_j}\sum_{i=1}^{K}e^{-\alpha x_i}}{\left(\sum_{i=1}^{K}e^{-\alpha x_i}\right)^2}
		+\frac{\alpha e^{-\alpha x_j}\sum_{i=1}^{K}x_i e^{-\alpha x_i}}{\left(\sum_{i=1}^{K}e^{-\alpha x_i}\right)^2}
\end{align*}
\endgroup
where $j=1,2,\dots,n$. It's evident that $\frac{\partial H_\alpha (x)}{\partial x_j}$ is non-negative when $\alpha{x_i}\leq1,\forall i$ is satisfied. Thus $H_\alpha(x)$ is non-decreasing.
As for the convexity, we compute the Hessian matrix of $H_\alpha (x)$ defined as $F \triangleq \nabla^2 H_\alpha (x)$, and next show that $v^T Fv$ is non-positive for any non-zero column vector $v=[v_1,v_2,\dots,v_n]^T$. More intuitively, we take its opposite form
\begingroup
\smaller
\begin{align*}
	-{v^T}Fv = A_1 - 2A_2 + A_3 - 2A_4   
\end{align*}
\endgroup 
where
\begingroup
\smaller
\begin{align*}
	A_1 &= \left(\sum_{i=1}^{K}e^{-\alpha x_i}\right)^{-1} \alpha \left(\sum_{i=1}^{K}(2-\alpha x_i)e^{-\alpha x_i}v_i^2\right) \notag \\
	A_2 &= \left(\sum_{i=1}^{K}e^{-\alpha x_i}\right)^{-2} \alpha \left(\sum_{i=1}^{K}(1-\alpha x_i)e^{-\alpha x_i}v_i\right) \left(\sum_{i=1}^{K}e^{-\alpha x_i}v_i\right) \notag \\
	A_3 &= \left(\sum_{i=1}^{K}e^{-\alpha x_i}\right)^{-2} \alpha^2 \left(\sum_{i=1}^{K}x_i e^{-\alpha x_i}\right) \left(\sum_{i=1}^{K}e^{-\alpha x_i}v_i^2\right) \notag \\
	A_4 &= \left(\sum_{i=1}^{K}e^{-\alpha x_i}\right)^{-3} \alpha^2 \left(\sum_{i=1}^{K}x_i e^{-\alpha x_i}\right) \left(\sum_{i=1}^{K}e^{-\alpha x_i}v_i^2\right).
\end{align*}
\endgroup
By applying the AM-GM inequality and the Cauchy-Schwarz inequality, it can be verified that
\allowdisplaybreaks
\begingroup  
\smaller  
	\begin{align*}  
		&A_1 - 2A_2\\  
		&\geq \alpha \left( \sum_{i=1}^{K} e^{-\alpha x_i} \right)^{-2} \left\{ \left(\sum_{i=1}^{K} e^{-\alpha x_i}\right) \left(\sum_{i=1}^{K} (2 - \alpha x_i) e^{-\alpha x_i} v_i^2\right) \right. \notag \\  
		&\left. \quad - \left[ \left(\sum_{i=1}^{K} e^{-\alpha x_i}\right) \left(\sum_{i=1}^{K} (2 - \alpha x_i) e^{-\alpha x_i} v_i^2\right) \right. \right. \notag \\  
		&\left. \left. \quad - \alpha \left(\sum_{i=1}^{K} x_i e^{-\alpha x_i}\right) \left(\sum_{i=1}^{K} e^{-\alpha x_i} v_i^2\right) \right] \right\} \notag \\  
		&= \alpha^2 \left( \sum_{i=1}^{K} e^{-\alpha x_i} \right)^{-3} \left(\sum_{i=1}^{K} x_i e^{-\alpha x_i}\right) \left(\sum_{i=1}^{K} e^{-\alpha x_i} v_i^2\right)\left(\sum_{i=1}^{K}e^{-\alpha x_i}\right) \notag \\  
		&\geq \alpha^2 \left( \sum_{i=1}^{K} e^{-\alpha x_i} \right)^{-3} \left(\sum_{i=1}^{K} x_i e^{-\alpha x_i}\right) \left(\sum_{i=1}^{K}e^{-\alpha x_i}v_i\right)^2  
	\end{align*}  
\endgroup

Similarly, we can verify that
\begingroup
\smaller
\begin{align*}  
		A_3 - 2A_4 
		\geq -\alpha^2 \left( \sum_{i=1}^{K} e^{-\alpha x_i} \right)^{-3} \left(\sum_{i=1}^{K} x_i e^{-\alpha x_i}\right) \left(\sum_{i=1}^{K}e^{-\alpha x_i}v_i\right)^2  
\end{align*}
\endgroup

Hence, we obtain $-v^T Fv \geq0$, which means the Hessian matrix of $H_{\alpha}(x)$ is negative semidefinite. As a result, $H_{\alpha}(x)$ is concave for $\alpha{x_i}\leq1,\forall i$.
\renewcommand{\appendixname}{Appendix~B}
\appendix[Proof of Lemma \ref{lemma2}]\label{appendix 2}
We rewrite $H_\alpha$$(x)$ as
\begingroup
\smaller
\begin{align*}
		H_\alpha (x)=\sum_{j=1}^{K}\frac{e^{-\alpha x_j}}{\sum_{i=1}^{K}e^{-\alpha x_i}}x_j=\sum_{j=1}^{K}\frac{1}{1+\sum_{i\neq j}e^{-\alpha \left(x_i - x_j\right)}}x_j
\end{align*}
\endgroup
where
\begingroup
\smaller
\begin{align*}
	\lim\limits_{\alpha \to +\infty}e^{-\alpha\left(x_i - x_j\right)}=\left\{\begin{array}{ll}0,&x_j<x_i\\1,&x_j=x_i\\\infty,&x_j>x_i\end{array}\right.
\end{align*}
\endgroup
Denote $x_{\min}=\min\left\{x_1,x_2,\dots,x_n\right\}$ and the number of $x_{\min}$ for $x_j$ as $N_{\min}$, then $H_\infty (x) = \sum_{j=1}^{N_{\min}}\frac{1}{1+\left(N_{\min} -1\right)}x_{\min} = \underset{j}\min x_j$.

\bibliographystyle{IEEEtran}
\bibliography{References-GC}
\end{document}